\documentclass[a4paper, 11pt]{article}
\usepackage{srcltx}
\usepackage{indentfirst}
\usepackage{graphicx}
\usepackage{overpic}
\usepackage{multirow}
\usepackage{hyperref}
\usepackage[hang]{subfigure}
\usepackage{amsmath, amssymb, amsthm}
\usepackage{color}
\usepackage[mathscr]{euscript}

\newtheorem{theorem}{Theorem}
\newtheorem{lemma}{Lemma}

\numberwithin{equation}{section}

\graphicspath{{images/}}

{\theoremstyle{remark} }

\title{Well-posedness of a non-local abstract Cauchy problem with a singular integral}

\author{
Haiyan Jiang 
\thanks{School of Mathematical Sciences, Beijing Institute of Technology, Beijing 100081, China. Email: {\tt hyjiang@math.bit.edu.cn}.},
~~Tiao Lu\thanks{CAPT, HEDPS, LMAM,
IFSA Collaborative Innovation Center of MoE, \& School of Mathematical Sciences, Peking University, Beijing 100871, China. Email: {\tt tlu@math.pku.edu.cn}.},
~~Xiangjiang Zhu\thanks{School of Mathematical Sciences, Peking University, Beijing 100871, China. Email: {\tt zxj709@pku.edu.cn}.}}

\begin{document}
\maketitle

\begin{abstract}
A non-local abstract Cauchy problem with a singular integral is studied, which is a closed system of two evolution equations for a real-valued function and a function-valued function. By proposing an appropriate Banach space, the well-posedness of the evolution system is proved under some boundedness and smoothness conditions on the coefficient functions. Furthermore, an isomorphism is established to extend the result to a partial integro-differential equation with singular convolution kernel, which is a generalized form of the stationary Wigner equation. Our investigation considerably improves the understanding of the open problem concerning the well-posedness of the stationary Wigner equation with inflow boundary conditions.

\vspace*{4mm}
\end{abstract}

\section{Introduction}
In this paper we consider the following initial value problem for the unknown functions $c(t)$ and $v(t,x)$,
\begin{equation}
\nonumber
  \left\{
  \begin{aligned}
  	& \frac{\mathrm{d} c(t)}{\mathrm{d} t} = h(t,0) c(t) + \displaystyle{\int_{-\infty}^{+\infty} K(t,-x') v(t,x') \mathrm{d}
	x'},\quad (t,x)\in [0,T]\times \mathbb{R}\\
	& \frac{\partial v(t,x)}{\partial t} = \frac{h(t,x)-h(t,0)}{x}c(t)
  + \\
  & \qquad\qquad +\int_{-\infty}^{+\infty} \displaystyle{\frac{ K(t,x-x') - K (t, -x')}{x}v(t,x')
  \mathrm{d} x '},\quad (t,x)\in [0,T]\times \mathbb{R} \\
  & c(0) = c_0\in \mathbb{R},\quad v(0,x) = v_0(x),\quad x\in \mathbb{R},
  \end{aligned}
  \right.
\end{equation}
where $h(t,x)$ and $K(t,x)$ are given real-valued functions. It can be rewritten as an abstract Cauchy problem 
\begin{equation}  \label{equ:1}
  \left\{
  \begin{aligned}
   &
  \frac{\mathrm{d}}{\mathrm{d}t}
  \begin{bmatrix}
    c(t) \\
    v(t)
  \end{bmatrix}=B(t)
  \begin{bmatrix}
    c(t) \\
    v(t)
  \end{bmatrix},\quad t\in [0,T] \\
  & 
  \begin{bmatrix}
    c(0) \\
    v(0)
  \end{bmatrix}=
  \begin{bmatrix}
    c_0 \\
    v_0
  \end{bmatrix},
  \end{aligned}
  \right.
\end{equation} 
where $v(t,x)$ is viewed as a vector-valued function of $t$, i.e., $v(t)=v(t,\cdot)$, and $B(t)$ is a linear operator,
\begin{equation}  \label{equ:2}
	B(t)  \begin{bmatrix}
	c \\ v \end{bmatrix}
  =  \begin{bmatrix}
	c h(t,0) + \displaystyle{\int_{-\infty}^{+\infty} K(t,-x)v(x)\mathrm{d} x} \\
	c \displaystyle{\frac{h(t,x)-h(t,0)}{x}} + \displaystyle{\int_{-\infty}^{+\infty}\frac{
  K(t,x-x')-K(t,-x')}{x}v(x')\mathrm{d} x' } \end{bmatrix}. 
\end{equation}
We will put forward an appropriate Banach space $X$ (See Section 2), on which $B(t)$ is a bounded linear operator under some boundedness and smoothness assumptions on $h$ and $K$. Therefore, the well-posedness of the abstract Cauchy problem \eqref{equ:1} is proved using the semigroup theory of linear evolution systems.

Meanwhile, the well-posedness result of \eqref{equ:1} is applied to initial value problem of the following partial integro-differential equation (PIDE),
\begin{equation}  \label{equ:4}
  \left\{
   \begin{aligned}
   \frac{\partial u(t,x)}{\partial t}&=\frac{1}{x}\Psi[K]u(t,x),\quad t\in [0,T], \\
   u(0,x)&=u_0(x)=\frac{c_0}{x} + v_0(x), \\
   \end{aligned}
   \right.
\end{equation}
where $\Psi[K]$ is the convolution operator with kernel $K$,
\begin{equation} 
\nonumber
  \Psi[K]u(t,x) = \int_{-\infty}^{+\infty} K(t,x-x')u(t,x')\mathrm{d} x'.
\end{equation}
The relation of problems \eqref{equ:1} and \eqref{equ:4} will be revealed and an isomorphism between their solutions is established (see Section 3). In this way, the well-posedness of \eqref{equ:4} can also be obtained applying previous analysis.

The PIDE appeared in \eqref{equ:4} is a generalized form of the stationary Wigner 
transport equation \cite{MRS90, Wigner1932}, which is a popular tool in the quantum transport simulation (especially in the nano semiconductor simulation). The one-dimensional stationary Wigner equation can be written as
\begin{equation} \label{eq:Wigner}
v \frac{\partial f(x,v)}{\partial x}  = \int_{-\infty}^{+\infty} V_{w}(x,v-v') f ( x, v')\mathrm{d} v',  
\end{equation}
where $f(x,v)$ is the quasi-probability density function in the phase
space $(x,v)$, and the Wigner potential $V_{w}(x,v)$ is related to 
the potential $V(x)$ through 
\begin{equation} \label{eq:Vw}
V_{w}(x,v) = \frac{\mathrm{i}}{2\pi} \int_{-\infty}^{+\infty} \mathrm{e}^{\mathrm{i} v y } \left[V(x+y/2)-V(x-y/2) \right] \mathrm{d} y .
\end{equation}
Eq. \eqref{eq:Wigner} with inflow boundary conditions
\begin{equation} \label{equ:6}
  f(0,v)=f_0(v)\text{ }(v>0),\quad f(1,v)=f_1(v)\text{ }(v<0)
\end{equation}
is widely used to obtain the current-voltage curve that is an important characteristic of semiconductor devices \cite{Frensley1987,Fre90,Tang2008}. 
However, the well-posedness of this inflow boundary value problem is still an open problem which has attracted the attention of many mathematicians and is only partly solved in \cite{ALZ00,Zweifel2001,Barletti2001,LiLuSun2014,LiLuSun2017}. One big issue is whether $L^2(\mathbb{R})$ is a suitable solution space for \eqref{eq:Wigner}. The well-posedness results previously explored enable us to some extent to investigate the stationary Wigner equation with inflow boundary conditions, which will be illustrated at the end of the paper. 

\section{Well-posedness of the Cauchy problem}

In this section, the well-posedness of the Cauchy problem \eqref{equ:1} is studied. For ease of understanding, we first state some results for the general evolution system \eqref{equ:7} (see e.g. Refs \cite{Pazy1983} and \cite{Fattorini2009}). Then we propose a proper solution space for the Cauchy problem \eqref{equ:1} and verify that the conditions required by these general results are fulfilled if sufficient smoothness and boundedness of the coefficient functions are assumed. In this way, the well-posedness of \eqref{equ:1} is proved as an application of the semigroup theory of linear operators

Let $X$ be a Banach space and $A(t)$: $D(A(t))\subset X\to X$ be a linear operator in X, $\forall t\in T$. Consider the initial value problem:
\begin{equation}  \label{equ:7}
  \left \{
  \begin{aligned}
    & \frac{\mathrm{d} u(t)}{\mathrm{d} t}=A(t)u(t),\quad 0\leq s< t\leq T, \\
    & u(s)=x\in X.
  \end{aligned}
  \right.
\end{equation}
A solution $u$ of \eqref{equ:7} is called a classical solution if $u\in C([0,T];X) \cap C^1( (0,T];X)$. Moreover, let $U(t,s)$ be the solution operator, which is a two-parameter family satisfying $u(t)=U(t,s)u(s)$. Then we have the following theorems.  
\begin{theorem}  \label{thm:1}
Let $X$ be a Banach space and let
$A(t)$ be a bounded linear operator on $X$ for every $t\in[0,T]$. If the function $t\to A(t)$ is continuous in the uniform operator topology, then the initial value problem \eqref{equ:7} has a unique
classical solution $u$.
\end{theorem}

\begin{theorem}  \label{thm:2}
Suppose that the conditions in Theorem \ref{thm:1} are satisfied.
Then $U(t, s)$ is bounded and continuous in the uniform operator topology for $0\leq s<t\leq T$. Moreover, 
\begin{equation}
  \|U(t, s)\| \leq \exp\left(\int_s^t \|A(\tau)\|\mathrm{d} \tau\right).
\end{equation}
\end{theorem}

Evidently, Theorems \ref{thm:1} and \ref{thm:2} describe the well-posedness of evolution system \eqref{equ:7} conclusively. Hence, the issue is to define a proper Banach space, say $X$, for the system \eqref{equ:1}, such that the operator $B(t)$ defined by \eqref{equ:2} satisfies all the conditions referring to $A(t)$ in Theorem \ref{thm:1}. In other words,
$B(t)$ is bounded on $X$ for every $t\in \mathbb{R}$ and is
continuous, as an operator-valued function of $t$, in the uniform operator
topology. 

In this paper, we assume that $X=\mathbb{R}\oplus L^2(\mathbb{R})$, where $(c,v)\in X$ if and only if $c\in \mathbb{R}$ and $v\in L^2(\mathbb{R})$, and the norm of $X$ is naturally defined by $\|(c,v)\|_X=\max(|c|,\|v\|_{L^2})$. For convenience of further discussion, we rewrite the problem \eqref{equ:1} in the following form
\begin{equation}  \label{equ:8}
  \left\{
  \begin{aligned}
   &
  \frac{\mathrm{d}}{\mathrm{d}t}
  \begin{bmatrix}
    c(t) \\
    v(t)
  \end{bmatrix}=
  \begin{bmatrix}
    B_1(t) & B_3(t) \\
    B_2(t) & B_4(t)
  \end{bmatrix}
  \begin{bmatrix}
    c(t) \\
    v(t)
  \end{bmatrix},\quad t\in [0,T], \\
  & 
  \begin{bmatrix}
    c(0) \\
    v(0)
  \end{bmatrix}=
  \begin{bmatrix}
    c_0 \\
    v_0
  \end{bmatrix}\in X,
  \end{aligned}
  \right.
\end{equation}
where
\begin{align}
  \nonumber
  & B_1(t):\mathbb{R}\rightarrow \mathbb{R},\quad c\mapsto ch(t,0), \\
  \nonumber
  & B_2(t):\mathbb{R}\rightarrow L^2(\mathbb{R}),\quad c\mapsto 
  c\cdot\frac{h(t,x)-h(t,0)}{x}, \\
  \nonumber  
  & B_3(t):L^2(\mathbb{R})\rightarrow \mathbb{R},\quad w(x)\mapsto
  \int_{-\infty}^{+\infty} K(t,-x)w(x)\mathrm{d}x, 
\end{align}
and $B_4(t):L^2(\mathbb{R})\rightarrow L^2(\mathbb{R})$,
\begin{equation}  \label{equ:9}
  B_4(t)w(x)=\int_{-\infty}^{+\infty}\frac{K(t,x-x')-K(t,-x')}{x}
  w(x')\mathrm{d} x'.
\end{equation}
In order to study the operator $B(t)$, we will first prove that for every $t\in [0,T]$, $B_i(t)$ ($i=1,2,3,4$) is bounded and continuous in the uniform operator topology under some smoothness and boundedness assumptions on the coefficient functions $h(t,x)$ and $K(t,x)$.

We begin with the boundedness discussion of $B_i(t)$. Note that $B_1(t):c\mapsto ch(t,0)$ solely depends on $h(t,0)$ and is thus bounded as long as $h(t,0)$ is well defined. (This is generally not true in a measurable function since changing the value of a single point will not alter the function itself.) To achieve this, we assume that $h(t,\cdot)$ is continuous at a neighborhood of zero point. Moreover, we assume $K(t,\cdot)\in L^2(\mathbb{R})$ for every $t\in[0,T]$. By the Cauchy-Schwarz inequality, we have
\begin{equation} \label{equ:9.5}
\int_{-\infty}^{+\infty} K(t,-x)w(x)\mathrm{d}x \leq \|\widetilde{K}(t,\cdot)\|_{L^2}\cdot \|w(\cdot)\|_{L^2}=\|{K}(t,\cdot)\|_{L^2}\cdot \|w(\cdot)\|_{L^2},
\end{equation}
where $\widetilde{K}(t,x)=K(t,-x)$. Hence $B_3(t)$ is also bounded and $\|B_3(t)\|\leq\|K(t,\cdot)\|_{L^2}$. For simplicity, in what follows we will express $B_3(t)$ in terms of the inner product,
\begin{equation} 
\nonumber
  B_3(t):w(x)\mapsto\int_{-\infty}^{+\infty} K(t,-x)w(x)\mathrm{d}x=\langle \widetilde{K}(t,\cdot),w(\cdot)\rangle.
\end{equation} 
The boundedness of $B_2(t)$ and $B_4(t)$ is established by the following lemmas.
\begin{lemma}  \label{lem:1}
  Suppose $w(x)$ is Lipschitz continuous
  with a Lipschitz constant $M_1$ and is bounded in terms of $\|w(
  \cdot)\|_\infty\leq M_2$. Then $\frac{w(x)-w(0)}{x}\in 
  L^2(\mathbb{R})$ and
\begin{equation}
  \nonumber
  \left \|\frac{w(x)-w(0)}{x}\right \|_{L^2}\leq \sqrt{2M_1^2+
  8M_2^2}.
\end{equation}
\end{lemma}

\begin{proof}
\begin{equation}
  \nonumber
  \begin{split}
    \left \|\frac{w(x)-w(0)}{x}\right \|_{L^2}^2&=\int_\mathbb{R}\left |\frac{w(x)-w(0)}{x}\right |^2\mathrm{d} x \\
    &=\int_{|x|\leq 1}\left |\frac{w(x)-w(0)}{x}\right |^2\mathrm{d} x+\int_{|x|> 1}\left |\frac{w(x)-w(0)}{x}\right |^2\mathrm{d} x \\
    &\leq \int_{|x|\leq 1}|M_1|^2\mathrm{d} x+\int_{|x|> 1}\left |\frac{2M_2}{x}\right |^2\mathrm{d} x \\
    &=2M_1^2+8M_2^2,
  \end{split}
\end{equation}
\end{proof}

The boundedness of $B_2(t)$ then follows by setting $w(\cdot)=h(t,\cdot)$ in Lemma \ref{lem:1}. To proceed, let $D^1 w$ denote the first-order weak derivative with respect to the variable $x$, where $w$ can be a real function of one or two variables, i.e., $w=w(x)$ or $w=w(t,x)$.  

\begin{lemma}  \label{lem:2}
Suppose $K(t,\cdot)\in H^{1}(\mathbb{R})=\{f(v)\in L^2(\mathbb{R}):D^1f\in L^2(\mathbb{R})\}$. Then $B_4(t)$ is a bounded linear operator on $L^2(\mathbb{R})$ and the corresponding operator norm is controlled by
\begin{equation}
  \nonumber
  \|B_4(t)\|\leq 2\sqrt{2}\|K(t,\cdot)\|_{H^1}.
\end{equation}
\end{lemma}

\begin{proof}
  Note that
\begin{equation}
  \nonumber
  B_4(t)u(t,x)=\frac{(K*u)(t,x)-(K*u)(t,0)}{x}.
\end{equation}
Thus, in light of Lemma \ref{lem:1}, we will investigate the Lipschitz continuity and boundedness of the function $(K*u)(t,\cdot)$. Using the Fubini theorem and the Cauchy-Schwarz inequality, we obtain
\begin{equation}
  \nonumber
  \begin{split}
    |(K*u)(t,x_2) & -(K*u)(t,x_1)| \\
    &=\left |\int_{-\infty}^{+\infty}[K(t,x_2-x')-K(t,x_1-x')]u(t,x')\mathrm{d} x'\right |\\
    &=\left |\int_{-\infty}^{+\infty}\left [\int_{x_1}^{x_2}D^1 K(t,\bar{x}-x')\mathrm{d} \bar{x}\right ]u(t,x')\mathrm{d} x'\right |\\
    &\leq\int_{-\infty}^{+\infty}\int_{x_1}^{x_2}\left |D^1 K(t,\bar{x}-x')\right |\cdot\left |u(t,x')\right |\mathrm{d} \bar{x}\mathrm{d} x'\\
    &=\int_{x_1}^{x_2}\left [\int_{-\infty}^{+\infty}\left |D^1 K(t,\bar{x}-x')\right |\cdot|u(t,x')|\mathrm{d} x'\right ]\mathrm{d}\bar{x}\\
    &\leq\int_{x_1}^{x_2}\left \| D^1 K(t,\bar{x}-\cdot)\right \|_{L^2}\cdot \left \|u(t,\cdot)\right \|_{L^2}\mathrm{d} \bar{x}\\
    &=\left \| D^1 K(t,\cdot)\right \|_{L^2}\cdot \left \|u(t,\cdot)\right \|_{L^2}\cdot|x_2-x_1|. 
  \end{split}
\end{equation}
On the other hand,
\begin{equation}
  \nonumber
  \begin{split}
    |(K*u)(t,x)|&\leq\int_{\mathbb{R}}|K(t,x-x')u(t,x')|\mathrm{d} x'\\
    &\leq\left (\int_\mathbb{R}|K(t,x-x')|^2\mathrm{d} x'\right )^{\frac{1}{2}}\cdot \|u(t,\cdot)\|_{L^2} \\
    &=\|K(t,\cdot)\|_{L^2} \|u(t,\cdot)\|_{L^2}.
  \end{split}
\end{equation}
Since $K(t,\cdot)\in H^{1}(\mathbb{R})$, replacing $w(\cdot)$ by $(K*u)(t,\cdot)$ in Lemma \ref{lem:1}, we conclude that $B_4(t)u(t,\cdot)\in L^2(\mathbb{R})$ and
\begin{equation}
  \nonumber
  \begin{split}
    \|B_4(t)u(t,\cdot)\|_{L^2} &=\left \|\frac{(K*u)(t,x)-(K*u)(t,0)}{x}\right \|_{L^2}\\
    &\leq \sqrt{2\|D^1K(t,\cdot)\|_{L^2}^2\|u(t,\cdot)\|_{L^2}^2 +8\|K(t,\cdot)\|_{L^2}^2\|u(t,\cdot)\|_{L^2}^2\cdot} \\
    &\leq 2\sqrt{2}\|K(t,\cdot)\|_{H^1} \|u(t,\cdot)\|_{L^2},
  \end{split}
\end{equation}
which demonstrates that $B_4(t)$ is a bounded linear operator on $L^2(\mathbb{R})$ and the operator norm is controlled by $2\sqrt{2}\|K(t,\cdot)\|_{H^1}$.
\end{proof}

Now we consider the continuity of $B_i(t)$ in the uniform operator topology (we will simply refer to continuity below). For sake of brevity, the notation $w'(t,x)$ is used to denote the partial derivative $\partial w(t,x)/\partial t$, where $w(t,x)$ is any two variable function defined on $[0,T]\times\mathbb{R}$. From our point of view, $w$ is considered as a vector-valued function with respect to $t$ and the notation is thus similar to that of the derivative in a real-valued function. In what follows, we assume that
\vspace{3mm}
\newline
{\rm (i)} For every $x\in\mathbb{R}$, $h(\cdot,x)$ is absolutely 
continuous on the interval $[0,T]$, i.e., 
\begin{equation} 
  \nonumber
  h(t_2,x)-h(t_1,x)=\int_{t_1}^{t_2}h'(t,x)\mathrm{d} t,\quad 0\leq t_1\leq t_2\leq T.
\end{equation}
\vspace{-2mm}
\newline
{\rm (ii)} Moreover, $h'(t,\cdot)$ is uniformly bounded in $L^{\infty}(\mathbb{R})$, i.e., 
\begin{equation}
  \nonumber
  \|h'(t,\cdot)\|_\infty \leq k_1,\quad \forall t\in [0,T],
\end{equation}
and Lipschitz continuous with a uniform Lipschitz constant, i.e., there exists $k_2>0$ such that
\begin{equation}
  \nonumber
  |h'(t,x_2)-h'(t,x_1)| \leq k_2|x_2-x_1|,\quad \forall t\in [0,T].
\end{equation}
\vspace{-3mm}
\newline
Obviously, setting $x=0$ in {\rm (i)} gives the continuity of $B_1(t)$. The continuity of $B_2(t)$ can be derived via
\begin{equation}
  \nonumber
  \begin{split}
  \left \|\frac{h(t_2,x)-h(t_2,0)}{x}-\frac{h(t_1,x)-h(t_1,0)}{x}\right \|_{L^2} & =\left \|\int_{t_1}^{t_2}\frac{h'(t,x)-h'(t,0)}{x}\mathrm{d} t\right \|_{L^2} \\
  & \leq \int_{t_1}^{t_2}\left \|\frac{h'(t,x)-h'(t,0)}{x}\right \|_{L^2}\mathrm{d} t \\
  & \leq \sqrt{2k_2^2+8k_1^2}\cdot|t_2-t_1|.
  \end{split}
\end{equation}
The last inequality comes directly from Lemma \ref{lem:1}.

For the continuity of $B_3(t)$ and $B_4(t)$, we assume
\vspace{3mm}
\newline
{\rm (iii)} $K(\cdot,x)$ is absolutely continuous for $t\in[0,T]$. (Thus $K'(t,\cdot)$ exists almost everywhere in $[0,T]$.)
\vspace{3mm}
\newline
{\rm (iv)} $K'(t,\cdot)$ is uniformly bounded on $H^1(\mathbb{R})$, i.e., there exists a constant $m$ such that
\begin{equation} \label{equ:9.6}
  \|K'(t,\cdot)\|_{H^1}\leq m, \quad \forall t \in [0,T].
\end{equation}
\vspace{-3mm}
\newline
Applying {\rm (iii)} we have, for all $0\leq t_1< t_2 \leq T$,
\begin{equation}
  \nonumber
  \begin{split}
    [B_4(t_2)- & B_4(t_1)]u(x) \\ 
    & =\int_{-\infty}^{+\infty}\left [\frac{K(t_2,x-x')-K(t_2,-x')}{x}-\frac{K(t_1,x-x')-K(t_1,-x')}{x}\right ]u(x')\mathrm{d} x' \\
    & =\int_{-\infty}^{+\infty}\left [\int_{t_1}^{t_2}K'(t,x-x')-K'(t,-x') \mathrm{d} t\right ]\frac{u(x')}{x}\mathrm{d} x' \\
    & =\int_{t_1}^{t_2}\left [\int_{-\infty}^{+\infty}\frac{K'(t,x-x')-K'(t,-x')}{x}u(x')\mathrm{d} x'\right ]\mathrm{d} t
  \end{split}
\end{equation}
Define $B'_4(t):L^2(\mathbb{R})\to L^2(\mathbb{R})$, 
\begin{equation}
  \nonumber
  B'_4(t)u(x)=\int_{-\infty}^{+\infty}\frac{K'(t,x-x')-K'(t,-x')}{x}u(x')\mathrm{d} x',\quad \forall u\in L^2(\mathbb{R}).
\end{equation}
We can thus write
\begin{equation}
  \nonumber
  [B_4(t_2)- B_4(t_1)]u(x)=\int_{t_1}^{t_2}\left [B'_4(t)u(x)\right ]\mathrm{d} t.
\end{equation}
Replacing $K(t,\cdot)$ by $K'(t,\cdot)$ in \eqref{equ:9} and Lemma \ref{lem:2} and using condition {\rm (iv)} (see \eqref{equ:9.6}) we know that $B_4'(t)$ is also a bounded linear operator on $L^2(\mathbb{R})$ and
\begin{equation}
  \nonumber
  \|B'_4(t)\|\leq 2\sqrt{2}\|K'(t,\cdot)\|_{H^1}\leq 2\sqrt{2}m,\quad \forall t\in[0,T].
\end{equation}
Thus we obtain
\begin{equation}
  \nonumber
  \begin{split}
    \left \|[B_4(t_2)-B_4(t_1)]\right \| & =\sup_{u\in L^2(\mathbb{R})}\frac{\left \|\int_{t_1}^{t_2}\left [B'_4(t)u(\cdot)\right ]\mathrm{d} t\right \|_{L^2}}{\|u(\cdot)\|_{L^2}}\leq \sup_{u\in L^2(\mathbb{R})}\frac{\int_{t_1}^{t_2}\left \|B'_4(t)u(\cdot)\right \|_{L^2}\mathrm{d} t}{\|u(\cdot)\|_{L^2}} \\
    & \leq \sup_{u\in L^2(\mathbb{R})}\frac{2\sqrt{2}m\|u(\cdot)\|_{L^2}|t_2-t_1|}{\|u(\cdot)\|_{L^2}}= 2\sqrt{2}m|t_2-t_1|.
  \end{split}
\end{equation}
On the other hand, by the Cauchy-Schwarz inequality,
\begin{equation}
  \nonumber
  \begin{split}
    \|B_3(t_2)-B_3(t_1)\| & =\sup_{u\in L^2(\mathbb{R})} \frac{\left|\left\langle\widetilde{K}(t_2,\cdot)-\widetilde{K}(t_1,\cdot),u(\cdot)\right\rangle\right|}{\|u(\cdot)\|_{L^2}} \leq \left\|\widetilde{K}(t_2,\cdot)-\widetilde{K}(t_1,\cdot)\right\|_{L^2} \\
    & =\left \|\int_{t_1}^{t_2}\widetilde{K}'(t,\cdot)\mathrm{d} t\right \|_{L^2} \leq \int_{t_1}^{t_2}\left \|\widetilde{K}'(t,\cdot)\right \|_{L^2}\mathrm{d} t \\
    & \leq \|K'(t,\cdot)\|_{H_1}|t_2-t_1|\leq 2\sqrt{2}m|t_2-t_1|.
  \end{split}
\end{equation}
Hence the continuity of $B_3(t)$ and $B_4(t)$ is also proved. 

Collecting all the previous results on $B_i(t)$, we can conclude the following theorem for the Cauchy problem \eqref{equ:1}.
\begin{theorem}  \label{thm:4}
Suppose that $\forall t\in [0,T]$, $K(t,\cdot)\in H^1(\mathbb{R})$ and $h(t,\cdot)\in L^{\infty}(\mathbb{R})$ is Lipschitz continuous on $\mathbb{R}$ with the minimal Lipschitz constant $L(t)$. Moreover, suppose the assumptions {\rm (i)}-{\rm (iv)} hold. Then the Cauchy problem \eqref{equ:1} has a unique classical solution in $C([0,T];X)$ and we have the following estimation
\begin{equation}
\nonumber
  \left\|
  \begin{bmatrix}
    c(t) \\
    v(t)
  \end{bmatrix}
  \right\|_X \leq 
  \left\|
  \begin{bmatrix}
    c_0 \\
    v_0
  \end{bmatrix}
  \right\|_X
  \cdot\exp\left(\int_0^t\left[ \sqrt{2L(s)^2+8\|h(s,\cdot)\|_{L^\infty}^2}+2\sqrt{2}\|K(s,\cdot)\|_{H^1}\right]\mathrm{d} s\right).
\end{equation}
\end{theorem}
\begin{proof}
For simplicity we denote $v(t)=v(t,x)$ and $\mathrm{d}v(t)/\mathrm{d}t=\partial v(t,x)/\partial t$. From the discussion on boundedness (see \eqref{equ:9.5}, Lemma \ref{lem:1} and Lemma \ref{lem:2}), we know that
\begin{equation}
\nonumber
  \begin{split}
    & \|B_1(t)\|\leq |h(t,0)|,\quad \|B_2(t)\|\leq\sqrt{2L(t)^2+8\|h(t,\cdot)\|_{L^\infty}^2},\quad \\
    & \|B_3(t)\|\leq\|K(t,\cdot)\|_{L^2},\quad \|B_4(t)\|\leq 2\sqrt{2}\|K(t,\cdot)\|_{H^1},
  \end{split}
\end{equation}
where $\|\cdot\|$ denotes the corresponding operator norm, respectively. Substituting it into \eqref{equ:8} we obtain (note that $|h(t,0)|\leq\|h(t,\cdot)\|_{L^\infty}$ and $K(t,\cdot)\|_{L^2}\leq \|K(t,\cdot)\|_{H^1}$)
\begin{equation}
\nonumber
  \begin{split}
  \left\|B(t)
  \begin{bmatrix}
    c(t) \\
    v(t)
  \end{bmatrix}
  \right\|_X & =\max\left(|B_1(t)c(t)+B_3(t)v(t)|,\|B_2(t)c(t)+B_4(t)v(t)\|_{L^2}\right) \\
  & \leq \max(\|B_1(t)\|+\|B_3(t)\|,\|B_2(t)\|+\|B_4(t)\|)\cdot\max(|c(t)|,\|v(t)\|_{L^2}) \\
  & = \left[ \sqrt{2L(t)^2+8\|h(t,\cdot)\|_{L^\infty}^2}+2\sqrt{2}\|K(t,\cdot)\|_{H^1}\right]\cdot
  \left\|
  \begin{bmatrix}
    c(t) \\
    v(t)
  \end{bmatrix}
  \right\|_X.
  \end{split}
\end{equation}
Thus we see that $B(t)$ is bounded on $X$ for every $t\in[0,T]$ and
\begin{equation}
\nonumber
  \|B(t)\|\leq \sqrt{2L(t)^2+8\|h(t,\cdot)\|_{L^\infty}^2}+2\sqrt{2}\|K(t,\cdot)\|_{H^1}.
\end{equation}    
Similarly, given the assumptions {\rm (i)}-{\rm (iv)} we can also verify the continuity of $B(t)$ (in the uniform operator topology). Hence applying Theorems \ref{thm:1} and \ref{thm:2} the proof is completed immediately.
\end{proof}

\section{Application to the partial integro-differential equation}

In this section, we will reveal the relation between the Cauchy problem \eqref{equ:1} and the PIDE \eqref{equ:4} and apply the previous results to the latter problem. Throughout this section, $h(t,x)$ is determined by
\begin{equation} \label{equ:9.8}
  h(t,x) = \int_{-\infty}^{+\infty} \frac{ K(t,x-x')}{x'} \mathrm{d} x'.
\end{equation} 
For simplicity, we will also use the notation $\Psi[K]u=K*u$ and $h=K*(1/x)$ in the following. To improve understanding, we state the core theorem first.

\begin{theorem} \label{thm:5}
Assume $\forall t\in [0,T]$, $K(t,\cdot)\in L^2(\mathbb{R})$ is Lipschitz continuous. If the PIDE \eqref{equ:4} has a solution $u$ in terms of $u(t,x)=c(t)/x+v(t,x)$, where $v(t,\cdot)\in L^2(\mathbb{R})$ for all $t\in [0,T]$, then $(c(t),v(t,x))$ is also a solution of the Cauchy problem \eqref{equ:1}. Conversely, if $(c(t),v(t,x))$ is the solution of \eqref{equ:1} and $v(t,\cdot)\in L^2(\mathbb{R})$, $\forall t\in [0,T]$, then $u(t,x)=c(t)/x+v(t,x)$ is a solution of the PIDE \eqref{equ:4}.
\end{theorem}

To begin with, we define the space
\begin{equation}
  \nonumber
  \mathcal{L}^2(\mathbb{R})=\left\{ f(x)=\frac{c}{x}+v(x):c\in\mathbb{R},v(\cdot)\in L^2(\mathbb{R}) \right\},
\end{equation}
which can be viewed as an one-dimensional extension of $L^2(\mathbb{R})$ since $1/x\notin L^2(\mathbb{R})$. We can thus equate $\mathcal{L}^2(\mathbb{R})$ with $X$ through an isomorphism:
\begin{equation}
\nonumber
  \sigma:\mathcal{L}^2(\mathbb{R})\to X,\quad f(x)=\frac{c}{x}+v(x)\mapsto (c,v(\cdot)).
\end{equation}
Accordingly, $\mathcal{L}^2(\mathbb{R})$ is a Banach space with the norm $\|f\|_{\mathcal{L}^2(\mathbb{R})}=\max(|c|,\|v\|_{L^2})$.

In the following, we will study the PIDE \eqref{equ:4} in the space $\mathcal{L}^2(\mathbb{R})$. Namely, we look for solutions in form of $u(t,\cdot)\in \mathcal{L}^2(\mathbb{R})$, i.e.,
\begin{equation} 
\nonumber
  u(t,x)=\frac{c(t)}{x}+v(t,x),\quad v(t,\cdot)\in L^2(\mathbb{R}),
\end{equation}
for which the convolution $\Psi[K]u=K*u$ is in the sense of the Cauchy principal integration,
\begin{equation} 
\nonumber
  \begin{split}
    K*u & := \lim_{\epsilon\to 0^+}\left[\int_{\epsilon}^{1/\epsilon}K(t,x-x')u(x')\mathrm{d} x'+\int_{-1/\epsilon}^{-\epsilon}K(t,x-x')u(x')\mathrm{d} x'\right] \\
    & = \int_0^{+\infty}\left [K(t,x-x')u(x')-K(t,x+x')u(-x')\right ]\mathrm{d} x'.
  \end{split}
\end{equation}
Particularly, we have
\begin{equation} \label{equ:11}
  \begin{split}
    h=K*\frac{1}{x} & = \int_0^{+\infty}\frac{K(t,x-x')-K(t,x+x')}{x'}\mathrm{d} x'.
  \end{split}
\end{equation}
We have the following lemmas concerning the properties of such integrals.

\begin{lemma} \label{lem:3}
Suppose $K(t,\cdot)\in L^2(\mathbb{R})$ is Lipschitz
continuous with a Lipschitz constant $k_1(t)$. Then
$h(t,\cdot)=K(t,\cdot)*(1/x)$ is well defined and
$$\|h(t,\cdot)\|_{L^\infty} \leq
2\left(k_1(t)+\|K(t,\cdot)\|_{L^2}\right).$$
\end{lemma}

\begin{proof}
According to the Lipschitz continuity,
\begin{equation}
  \nonumber
  \begin{split}
    |h(t,x)|&=\left |\int_0^{+\infty}\frac{K(t,x-x')-K(t,x+x')}{x'}\mathrm{d} x'\right | \\
    &\leq\left (\int_{0}^1+\int_{1}^{+\infty}\right )\left |\frac{K(t,x-x')-K(t,x+x')}{x'}\right |\mathrm{d} x' \\
    &\leq \int_{0}^1|2k_1(t)|\mathrm{d} x'+\int_{1}^{+\infty}\left |\frac{K(t,x-x')}{x'}\right |\mathrm{d} x'+\int_{1}^{+\infty}\left |\frac{K(t,x+x')}{x'}\right |\mathrm{d} x'.
  \end{split}
\end{equation}
Using the Cauchy-Schwarz inequality, we obtain
\begin{equation}
  \nonumber
  \begin{split}
    \int_{1}^{+\infty}\left |\frac{K(t,x-x')}{x'}\right |\mathrm{d} x'&\leq \left (\int_{1}^{+\infty} |K(t,x-x')|^2\mathrm{d} x'\right )^{\frac{1}{2}}\cdot\left (\int_{1}^{+\infty}\left |\frac{1}{x'}\right |^2\mathrm{d} x'\right )^{\frac{1}{2}} \\
    &\leq \|K(t,\cdot)\|_{L^2}.
  \end{split}
\end{equation}
Similarly, 
\begin{equation}
  \nonumber
  \int_{1}^{+\infty}\left |\frac{K(t,x+x')}{x'}\right |\mathrm{d} x'\leq \|K(t,\cdot)\|_{L^2}.
\end{equation}
Therefore the improper integral \eqref{equ:11} converges and
\begin{equation}
  \nonumber
  |h(t,\cdot)|\leq 2(k_1(t)+\|K(t,\cdot)\|_2).
\end{equation}
.
\end{proof}

\begin{lemma} \label{lem:4}
Suppose $D^1K(t,\cdot)\in L^2(\mathbb{R})$ ($D^1$ denotes the first-order weak derivative with respect to the variable $x$ as in the previous section) is Lipschitz
continuous with a Lipschitz constant $k_2(t)$. Then
$h(t,\cdot)=K(t,\cdot)*(1/x)$ is Lipschitz continuous with the Lipschitz
constant $2(k_2(t)+\|D^1K(t,\cdot)\|_2)$. 
\end{lemma}

\begin{proof}
Let $V_K(t,x,x')=K(t,x-x')-K(t,x+x')$ and $D^1V_K(t,x,x')$, $D^1K(t,x-x')$ be the first-order weak derivative with respect to the (second) variable $x$. Since $D^1K(t,\cdot)\in L^2(\mathbb{R})\subseteq L^1_{\mathrm{loc}}(\mathbb{R})$, we have $D^1V_K(t,\cdot,x')\in L^1_{\mathrm{loc}}(\mathbb{R})$ for any fixed $x'$ and thus $V_K(t,\cdot,x')$ is absolutely continuous, i.e., for $\forall x_1 < x_2$,
\begin{equation}
  \nonumber
  \begin{split}
  V_K(t,x_2,x')-V_K(t,x_1,x') & =\int_{x_1}^{x_2} D^1V_K(t,\bar{x},x')\mathrm{d} \bar{x}.
  \end{split}
\end{equation}
Therefore,
\begin{equation}
  \nonumber
  \begin{split}
    |h(t,x_2)-h(t,x_1)| & =\left |\int_0^{+\infty}\frac{V_K(t,x_2,x')-V_K(t,x_1,x')}{x'}\mathrm{d} x'\right | \\
    & =\left |\int_0^{+\infty}\left [\int_{x_1}^{x_2} D^1V_K(t,\bar{x},x')\mathrm{d} \bar{x}\right ]\frac{1}{x'}\mathrm{d} x'\right | \\
    & \leq\int_{x_1}^{x_2}\left [\int_0^{+\infty}\left | D^1V_K(t,\bar{x},x') \right |\left |\frac{1}{x'}\right |\mathrm{d} x'\right ]\mathrm{d} \bar{x} \\
    & =\int_{x_1}^{x_2}\left [\int_0^{+\infty}\left |\frac{D^1K(t,x-x')-D^1K(t,x+x')}{x'}\right |\mathrm{d} x'\right ]\mathrm{d} \bar{x} \\
    & =\int_{x_1}^{x_2}\left |\left(D^1K*\frac{1}{x}\right)(t,\bar{x})\right |\mathrm{d} \bar{x}.
  \end{split}
\end{equation}
Finally, replacing $K$ by $D^1K$ in Lemma \ref{lem:3}, we obtain
\begin{equation}
  \nonumber
  \begin{split}
    |h(t,x_2)-h(t,x_1)| & \leq \int_{x_1}^{x_2}\left [2(k_2(t)+\|D^1K(t,\cdot)\|_2)\right ]\mathrm{d} \bar{x} \\
    & =2(k_2(t)+\|D^1K(t,\cdot)\|_2)|x_2-x_1|.
  \end{split}
\end{equation}
\end{proof}

\begin{proof}[Proof of Theorem \ref{thm:5}]
Assume $u(t,\cdot)\in \mathcal{L}^2(\mathbb{R})$, $t\in[0,T]$, is a solution of the PIDE \eqref{equ:4}. Substituting it into
\eqref{equ:4}, we obtain
\begin{equation}  \label{equ:13}
  \frac{\mathrm{d}c(t)}{\mathrm{d}t}+x\frac{\partial v(t,x)}{\partial t}=c(t)\left(K(t,x)*\frac{1}{x}\right)+(K*v)(t,x).
\end{equation}
According to Lemma \ref{lem:3}, $h(t,x)=K(t,x)*(1/x)$ is well defined on $[0,T]\times\mathbb{R}$ and using Young's inequality,
\begin{equation}
  \nonumber
  \|(K*v)(t,\cdot)\|_{L^\infty}\leq \|K(t,\cdot)\|_{L^2}\|v(t,\cdot)\|_{L^2}.
\end{equation}
Thus $K*v$ is also well defined. Let $x=0$ in both sides of \eqref{equ:13} and note that $h=K*(1/x)$. We have
\begin{equation}  \label{equ:14}
  \frac{\mathrm{d}c(t)}{\mathrm{d}t}=c(t)h(t,0)+\int_{-\infty}^{+\infty} K(t,-x)v(x)\mathrm{d} x. 
\end{equation}
In order to describe the evolution of $v(t,x)$, we substitute \eqref{equ:14} into \eqref{equ:13} and obtain
\begin{equation}  \label{equ:15}
  \begin{split}
    \frac{\partial v(t,x)}{\partial t} & =\frac{1}{x}\left[c(t)h(t,x)+(K*v)(t,x)-\frac{\mathrm{d}c(t)}{\mathrm{d}t}\right] \\
    & =\frac{1}{x}\left[c(t)h(t,x)+(K*v)(t,x)-c(t)h(t,0)-\int_{-\infty}^{+\infty} K(t,-x)v(x)\mathrm{d} x\right] \\
    & = c(t)\frac{h(t,x)-h(t,0)}{x}+\int_{-\infty}^{+\infty}\frac{K(t,x-x')-K(t,-x')}{x}v(t,x')\mathrm{d} x'.
  \end{split}
\end{equation}
Apparently, the evolution system composed of \eqref{equ:14} and \eqref{equ:15} is exactly the Cauchy problem \eqref{equ:1}. The inverse proposition can also be verified by simply reversing the proof above.
\end{proof}

Theorem \ref{thm:4} reveals that the Cauchy problem \eqref{equ:1} and the PIDE \eqref{equ:4} are equivalent if $h$ is defined by \eqref{equ:9.8}. Hence, the well-posedness of the PIDE \eqref{equ:4} can be investigated as an application of the previous section (see Theorem \ref{thm:4}). In order to correspond to the assumptions in Theorem \ref{thm:4}, we shall propose some extra conditions on $K$.
\vspace{3mm}
\newline
{\rm (v)} Both $\partial K(t,\cdot)/\partial t$ and $D^1(\partial K(t,\cdot)/\partial t)$ are uniformly Lipschitz continuous and uniformly bounded in $L^2(\mathbb{R})$ for $t\in[0,T]$. Namely, there exists a constant $M$ independent of $t$, such that for every $t\in[0,T]$, the corresponding Lipschitz constants and $L^2$ norms are all less than $M$. 

\begin{theorem} \label{thm:6}
Let the convolution kernel $K$ satisfy {\rm (iii)}-{\rm (v)}. Moreover, let $K(t,\cdot)\in H^1(\mathbb{R})$ and $K(t,\cdot)$, $D^1K(t,\cdot)$ be Lipschitz continuous and $k_1(t)$, $k_2(t)$ are the minimal Lipschitz constants, respectively. Then the PIDE \eqref{equ:4} has a unique classical solution $u\in C([0,T];\mathcal{L}^2(\mathbb{R}))$ and 
\begin{equation}
\nonumber
  \|u(t,\cdot)\|_{\mathcal{L}^2(\mathbb{R}))}\leq\|u_0\|_{\mathcal{L}^2(\mathbb{R}))}\cdot \exp\left(\int_0^t 2\sqrt{2}\left(k_2(s)+2k_1(s)+3\|K(s,\cdot)\|_{H^1}\right)\mathrm{d} s\right).  
\end{equation}
\end{theorem}
\begin{proof}
Applying Lemmas \ref{lem:3} and \ref{lem:4}, we know that
\begin{equation} \label{equ:16}
  \|h(t,\cdot)\|_{L^\infty}\leq 2(k_1(t)+\|K(t,\cdot)\|_{L^2}),\quad \forall t\in[0,T],
\end{equation} 
and $h(t,\cdot)$ is Lipschitz continuous with the minimal Lipschitz constant
\begin{equation} \label{equ:17}
  L(t)\leq  2(k_2(t)+\|D^1K(t,\cdot)\|_{L^2}),\quad \forall t\in[0,T].
\end{equation}

In addition, replacing $K(t,x)$ by $K'(t,x)$ in Lemmas \ref{lem:3} and \ref{lem:4} and applying {\rm (v)}, we know that $[K'*(1/x)](t,\cdot)$ is well defined, uniformly bounded in $L^\infty(\mathbb{R})$ and uniformly Lipschitz continuous for $t\in[0,T]$. Meanwhile, since $K(\cdot,x)$ is absolutely continuous, for any $0\leq t_1<t_2\leq T$ and $x\in \mathbb{R}$, we have
\begin{equation}
  \nonumber
  \begin{split}
    h(t_2,x)- & h(t_1,x) \\ 
    & =\int_0^{+\infty}\left [\frac{K(t_2,x+x')-K(t_2,x-x')}{x'}-\frac{K(t_1,x+x')-K(t_1,x-x')}{x'}\right ]\mathrm{d} x' \\
    & =\int_0^{+\infty}\left [\int_{t_1}^{t_2}K'(t,x+x')-K'(t,x-x') \mathrm{d} t\right ]\frac{1}{x'}\mathrm{d} x' \\
    & =\int_{t_1}^{t_2}\left [\int_0^{+\infty}\frac{K'(t,x+x')-K'(t,x-x')}{x'}\mathrm{d} x'\right ]\mathrm{d} t\\
    & =\int_{t_1}^{t_2}\left [K'(t,x)*\frac{1}{x}\right ]\mathrm{d} t.
  \end{split}
\end{equation}
Therefore, $h(\cdot,x)$ is absolutely continuous and
\begin{equation}
  \nonumber
  h'(t,x)=K'(t,x)*\frac{1}{x}.
\end{equation}
Conditions {\rm (i)} and {\rm (ii)} then follows according to the previous conclusions on $K'*(1/x)$. 

Now, all the conditions referring to $h$ and $K$ in Theorems \ref{thm:4} and \ref{thm:5} are fulfilled. Hence the PIDE \eqref{equ:4} has a unique classical solution $u\in C([0,T];\mathcal{L}^2(\mathbb{R}))$ and using \eqref{equ:16}, \eqref{equ:17} and the estimation in Theorem \ref{thm:4} we obtain
\begin{equation}
\nonumber
  \begin{split}
    \|u(t,\cdot)\|_{\mathcal{L}^2(\mathbb{R}))} & \leq \|u_0\|_{\mathcal{L}^2(\mathbb{R}))}\cdot \exp\left(\int_0^t\left[ \sqrt{2L(s)^2+8\|h(s,\cdot)\|_{L^\infty}^2}+2\sqrt{2}\|K(s,\cdot)\|_{H^1}\right]\mathrm{d} s\right) \\
    & \leq \|u_0\|_{\mathcal{L}^2(\mathbb{R}))}\cdot \exp\left(\int_0^t 2\sqrt{2}\left(k_2(s)+2k_1(s)+3\|K(s,\cdot)\|_{H^1}\right)\mathrm{d} s\right). 
  \end{split}
\end{equation}
The last inequality holds due to
\begin{equation}
\nonumber
  \begin{split}
    & \sqrt{2L(s)^2+8\|h(s,\cdot)\|_{L^\infty}^2}+2\sqrt{2}\|K(s,\cdot)\|_{H^1} \\
    = & 2\sqrt{2}\left(\sqrt{(k_2(s)+\|D^1K(s,\cdot)\|_{L^2})^2+(2k_1(s)+2\|K(s,\cdot)\|_{L^2})^2}+\|K(s,\cdot)\|_{H^1}\right) \\
    \leq & 2\sqrt{2}\left(k_2(s)+\|D^1K(s,\cdot)\|_{L^2}+2k_1(s)+2\|K(s,\cdot)\|_{L^2}+\|K(s,\cdot)\|_{H^1}\right) \\
    \leq & 2\sqrt{2}\left(k_2(s)+2k_1(s)+3\|K(s,\cdot)\|_{H^1}\right).
  \end{split}
\end{equation}
\end{proof}

We will conclude this section by an illustrative example regarding the stationary Wigner equation \eqref{eq:Wigner}. First we look at the initial value problem 
\begin{equation} \label{equ:18}
  \left\{
  \begin{aligned}
    & v \frac{\partial f(x,v)}{\partial x}=\int_{-\infty}^{+\infty} V_{w}(x,v-v')f(x,v')\mathrm{d} v',\\
    & f(0,v)=f_0(v).
  \end{aligned}
  \right.
\end{equation} 
Evidently, this is a special case of the PIDE \eqref{equ:4} (simply replacing the variables $(t,x)$ by $(x,v)$ in the discussion above). We put $f_0(v)=v\exp(-v^2/2)$ and $V(x)=\exp(-x^2/2)$. By the relation \eqref{eq:Vw} the corresponding convolution kernel $V_w$ can be expressed analytically, 
\begin{equation}
\nonumber
 V_w(x,v)=\sqrt{8/\pi}\exp(-2v^2)\sin(2vx).
\end{equation} 
One can thus easily see that $V_w$ satisfies the conditions required by Theorem \ref{thm:6}. Hence our theory indicates that the solution has the following form,
\begin{equation}
\nonumber
  f(x,v)=\frac{c(x)}{v}+f_0(x,v),\quad f_0(x,\cdot)\in L^2(\mathbb{R}).
\end{equation}
This can be further verified by a numerical experiment (see Figure \ref{fig:1}). It seems our theory handles the initial value problem \eqref{equ:18} perfectly.
\begin{figure}[htbp] 
  \centering
  \includegraphics[width=0.49\linewidth]{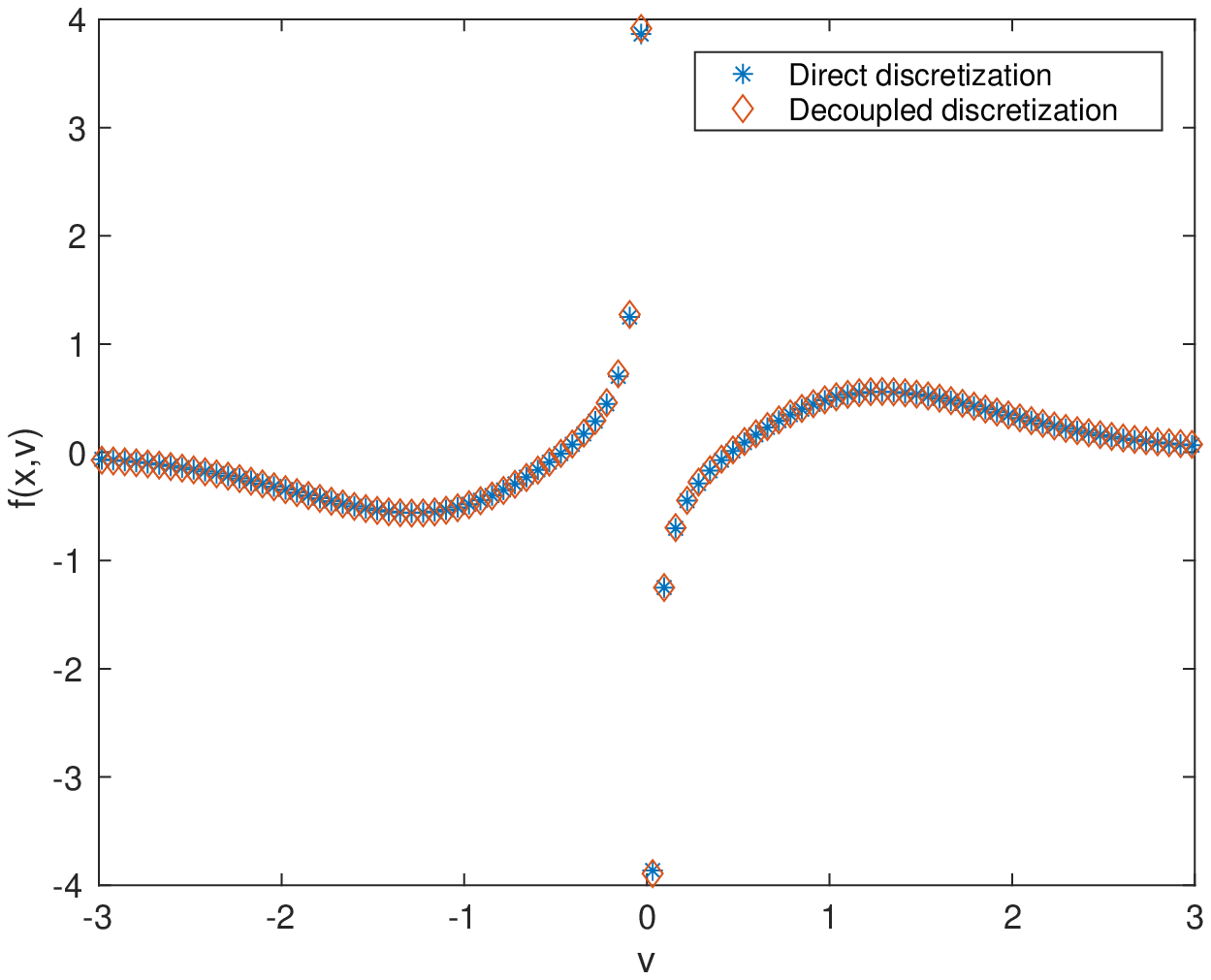}
  \includegraphics[width=0.49\linewidth]{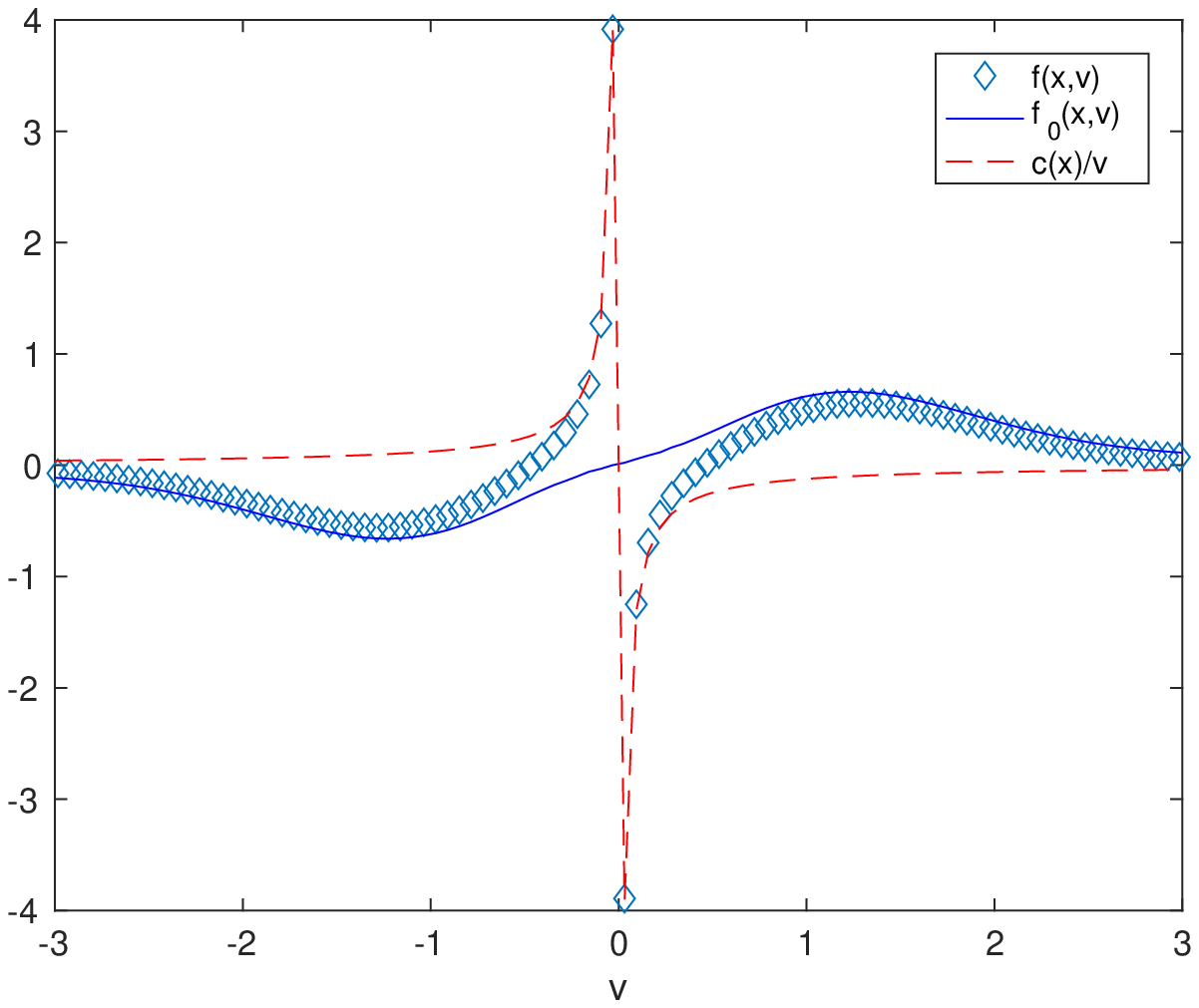}
  \caption{Numerical solutions of the IVP \eqref{equ:18} at $x=1$ using finite difference methods. The left picture displays the numerical result of $f(x,v)$ using a direct velocity discretization in \eqref{equ:18} combined with a RK4 (fourth order Runge-Kutta) method for the configuration propagation. Note that the zero velocity has been avoided (and should be avoided) in the direct discretization to ensure stability. Although the initial data is smooth enough, one can see that the singularity appears at the zero velocity while the solution propagates along the $x$-axis. On the right is the numerical solution of system \eqref{equ:1} with $h=V_w*(1/v)$, also using RK4 for the configuration propagation (one advantage here is that there is no need to avoid the zero velocity anymore). In this picture, however, the smooth part of the solution $f_0\in L^2(\mathbb{R})$ referred by the solid line and the singular part $c(x)/v$ referred by the dashed line can be further decoupled. One can see that the solution $f(x,v)$ approximates $c(x)/v$ near the zero velocity, while approximates $f_0$ in the region $|v|>1$. The numerical solution of \eqref{equ:1}, referred by the ``decoupled discretization'', is also compared to the direct discretization method of \eqref{equ:18} in the left picture.}
  \label{fig:1}
\end{figure}

The inflow boundary value problem \eqref{equ:6} can also be investigated through our theory absorbing the idea of parity decomposition. One significant aspect of the evolution operator $\Psi[V_w]/v$ ($\Psi[V_w]$ is the convolution operator as defined in \eqref{equ:4}) is the parity-preserving property (see e.g. \cite{Zweifel2001,Barletti2001,LiLuSun2017}), which states that an odd/even function $f_0(v)$ (see \eqref{equ:18}) will remain odd/even while propagating along the $x$-axis. This leads us to the parity decomposition
\begin{equation} \label{equ:19}
  f(x)=f_{o}(x)+f_{e}(x),\quad \forall f(x,\cdot)\in\mathcal{L}^2(\mathbb{R}),
\end{equation}
where $f_{o}$/$f_e$ is the odd/even part and we suppress the dependence on variable $v$ at times. Denote $\mathcal{L}^2_o(\mathbb{R})$/$\mathcal{L}^2_e(\mathbb{R})$ the set of odd/even functions in $\mathcal{L}^2(\mathbb{R})$. According to our theory and the parity-preserving property, the operator $\Psi[V_w]/v$ is a bounded operator on both $\mathcal{L}^2_o(\mathbb{R})$ and $\mathcal{L}^2_e(\mathbb{R})$. Moreover, since the singular part of $f(x)$ is always odd, we have $\mathcal{L}^2_e(\mathbb{R})$ equals to $L^2_e(\mathbb{R})$, the set of all even functions in $L^2(\mathbb{R})$. Therefore, we can define the solution operators $U_o(t,s)$ and $U_e(t,s)$, bounded on $\mathcal{L}^2_o(\mathbb{R})$ and $L^2_e(\mathbb{R})$ (see Theorem \ref{thm:2}) respectively, such that
\begin{equation} \label{equ:20}
  f_o(t)=U_o(t,s)f_o(s),\quad f_e(t)=U_e(t,s)f_e(s).
\end{equation}
To treat the inflow boundary value problem, for any real function $w(v)$, we define $w^+(v)$ the restriction of $w(v)$ on the half line $\mathbb{R}^+$ and $w^-(v)$ for the other half $\mathbb{R}^-$. We can thus write the inflow boundary conditions \eqref{equ:6} into
\begin{equation}
\nonumber
  f^+(0)=f_0^+,\quad f^-(1)=f_1^-,
\end{equation}
which, when substituting the decomposition \eqref{equ:19}, can be further interpreted by
\begin{equation} \label{equ:21}
  f^+_o(0)+f^+_e(0)=f_0^+,\quad -f^+_o(1)+f^+_e(1)=f_1^-.
\end{equation}
Meanwhile, since any odd/even function can be identified with its restriction on the half line $\mathbb{R}^+$, the solution operator $U_o(t,s)$/$U_e(t,s)$ induces a corresponding solution operator $U_o^+(t,s)$/$U_e^+(t,s)$ which, due to the fact proved on the former one, is also bounded in $\mathcal{L}^2_o(\mathbb{R}^+)$/$L^2_e(\mathbb{R}^+)$, where $\mathcal{L}^2_o(\mathbb{R}^+)$/$L^2_e(\mathbb{R}^+)$ is the $\mathcal{L}^2$/$L^2$ counterpart on the half line $\mathbb{R}^+$. Similar to \eqref{equ:20}, we have
\begin{equation} \label{equ:22}
  f_o^+(1)=U_o^+(1,0)f_o^+(0),\quad f_e^+(1)=U_e^+(1,0)f_e^+(0). 
\end{equation} 
Putting \eqref{equ:21} and \eqref{equ:22} together we obtain a closed system, for which the initial data can be solved formally,
\begin{equation}
\nonumber
  \begin{split}
    & f_o^+(0)=\left[U_o^+(1,0)+U_e^+(1,0)\right]^{-1}\left[U_e^+(1,0)f_0^+-f_1^-\right], \\
    & f_e^+(0)=\left[U_o^+(1,0)+U_e^+(1,0)\right]^{-1}\left[U_o^+(1,0)f_0^++f_1^-\right].
  \end{split}
\end{equation}
However, this is not the final answer to the inflow boundary value problem since we have no clue whether $\left[U_o^+(1,0)+U_e^+(1,0)\right]^{-1}$ is a bounded operator on a certain space (a similar result can be found in \cite{Barletti2001}).

\section{Conclusions}
We prove the well-posedness of the abstract Cauchy problem \eqref{equ:1} and extend our results to the study of the partial integro-differential Eq. \eqref{equ:4}, which is a generalized form of the stationary Wigner equation. Our theory reveals the equivalence of these two seemingly unrelated problems. Also, the results on the initial value problem \eqref{equ:4} shed lights on the inflow boundary value problem of the stationary Wigner equation, although a throughly solution of this problem definitely requires a further investigation. Another interesting question comes from the observation that if $c(t)\equiv 0$ for $t\in[0,T]$ (see Theorem \ref{thm:5}), then $u(t,x)$ is purely a solution of \eqref{equ:4} in $L^2(\mathbb{R})$. However, whether $L^2(\mathbb{R})$ is a proper solution space for the PIDE \eqref{equ:4} (or particularly the stationary Wigner equation \eqref{eq:Wigner} with initial or inflow boundary conditions) is not in the scope of this paper. These topics will be treated in future study.
\section*{Acknowledgments}
This research was supported in part by the NSFC (91434201, 91630130, 11671038, 11421101).


\end{document}